\documentclass[12pt,twoside]{article}
\usepackage{float}
\usepackage{amsfonts}
\usepackage{amssymb}
\usepackage{amsmath}
\usepackage{amsthm}
\usepackage{mathrsfs}		% for fancy \mathscr text style
\usepackage{stmaryrd}
\usepackage{listings}
\usepackage{multicol}
\usepackage{amsmath}
\usepackage[english]{babel}
\usepackage{graphicx}
\usepackage{color}
\usepackage{boxedminipage}
\usepackage{url}

\usepackage{tikz}
\usepackage{graphicx}
\usepackage{pgf}
\usepackage{tikz}
\usetikzlibrary{arrows,automata}
\usepackage[latin1]{inputenc}
\tikzstyle{every picture}+=[remember picture]

\usepackage[ruled,vlined,commentsnumbered]{algorithm2e}
\usepackage{pgf}
\usepackage{tikz}
\usetikzlibrary{matrix}
\usetikzlibrary{arrows,automata}
\usetikzlibrary{arrows,shapes}
\usepackage{float}
\usepackage{caption}
\usepackage{subcaption}
\usetikzlibrary{arrows,matrix,positioning}

\lstdefinestyle{customc}{
  belowcaptionskip=1\baselineskip,
  breaklines=true,
  frame=L,
  xleftmargin=\parindent,
  language=C,
  showstringspaces=false,
  basicstyle=\footnotesize\ttfamily,
  keywordstyle=\bfseries\color{green!40!black},
  commentstyle=\itshape\color{purple!40!black},
  identifierstyle=\color{blue},
  stringstyle=\color{orange},
}

\newcommand\fixmet[1]{}

\DeclareMathOperator{\A}{\textsf{A}}

\newtheorem{prop}{Proposition}[section]

\newtheorem{LM}{Lemma}[section]

\newtheorem{thm}{Theorem}[section]
\newtheorem{df}{Definition}[section]
\newtheorem{cor}{Corollary}[section]

\theoremstyle{pourlesremarques}
\newtheorem{rem}{Remark}[section]
\newtheorem{ex}{Example}[section]

\newcommand{\R}{\mathbb{R}}
\renewcommand{\l}{\lambda}
\newcommand{\C}{\mathbb{C}}
\newcommand{\Q}{\mathbb{Q}}
\renewcommand{\H}{\mathbb{H}}
\newcommand{\Z}{\mathbb{Z}}

\newcommand{\1}{\mathbf{1}}

\newcommand{\D}{\Delta}

\renewcommand{\i}{\mathbf{id}_E}
\newcommand{\f}{\mathbf{f}}
\newcommand{\xx}{\mathbf{x}}

\begin{document}

%\label{firstpage}

\title{Characterization of Termination for Linear Loop Programs}

\author{Rachid Rebiha\thanks{Instituto  de Computa\c{c}\~{a}o, Universidade Estadual  de Campinas, 13081\-970  Campinas,  SP.  Pesquisa desenvolvida com suporte financeiro da FAPESP, processos 2011\/08947\-1 e FAPESP BEPE 2013\/04734\-9 } \and  Arnaldo Vieira Moura\thanks{Instituto  de Computa\c{c}\~{a}o, Universidade Estadual  de Campinas, 13081\-970  Campinas,  SP.} \and Nadir Matringe \thanks{ Universit\'{e} de Poitiers, Laboratoire Math\'{e}matiques et Applications and Institue de Mathematiques de Jussieu Universit\'{e} Paris 7-Denis Diderot, France.}}

%\date{}

\maketitle
%\maketitle
%%\doulespacing
%\input{abstractSAS2013}

\begin{abstract}
%To the best of our knowledge, 
We present necessary and sufficient conditions for the termination of linear homogeneous programs. %homogeneous programs.% with one loop condition. 
 We also develop a complete method to check termination for this class of programs. Our complete characterization of termination for such programs is based on linear algebraic methods. We reduce the verification of the termination problem to checking the orthogonality of a well determined vector space and a certain vector, both related to loops in the program. Moreover, we provide theoretical results and symbolic computational methods guaranteeing the soundness, completeness and numerical stability of the approach. 
%while restricting  variable interpretations over a specific countable subring of $\R^n$. 
Finally, we show that it is enough to interpret variable values over a specific countable number field, or even over its ring of integers, when one wants to check termination over the reals.

\end{abstract}

% !TeX spellcheck=en_US
\section{Introduction}\label{intro}
%FIME: I commented out this paragraph in order to save space:
%Research on formal methods for program verification \cite{77cousot, manna2, mc-clarke, mc-sifakis} aim at discovering mathematical techniques and developing
%their associated algorithms to establish the correctness of software, hardware, concurrent systems, embedded systems or hybrid systems.  
Static program analysis \cite{Cousot, manna2, CousotHalbwachs78-POPL} is used to check that a software is free of defects, such as buffer overflows or segmentation faults, which are safety properties, or termination, which is a liveness property. 
Verification of temporal properties of infinite state systems \cite{Sipma-1999} is another example.  
Proving termination of $\textsf{while}$ loop programs is necessary for the verification of liveness properties that any well behaved and engineered system, or any safety critical embedded system, must guarantee. 
We could 
list here many verification approaches that are only practical depending on the facility with which termination can be automatically determined. 
More recent work on automated termination analysis of imperative loop programs has focused on partial decision procedures based on the discovery and synthesis of ranking functions.
Such functions map the loop variable to a well-defined domain where their value decreases at each iteration of the loop \cite{Colon:2001, Colon02}. 
Several interesting approaches, based on the generation of \emph{linear} ranking functions, have been proposed \cite{Bradley05linearranking, Bradley05terminationanalysis} for loop programs where the guards and the instructions can be expressed in a logic supporting linear arithmetic. 
For the generation of such functions there are effective heuristics \cite{Dams,Colon02} and, in some cases, there are also complete methods \cite{Podelski04}.  
On the other hand, it is easy to generate a simple linear terminating loop program that does not have a linear ranking function. 
In these cases, complete synthesis methods \cite{Podelski04} fail to provide a conclusion 
about the termination or nontermination of such programs.  

In this work we are motivated by the termination problem for linear $\textsf{while}$  loop programs. %In other words, 
In this class of loop programs, the loop condition is a conjunction of linear inequalities and the assignments to each of the variables in the loop instruction block are of an affine or linear form. 
In matrix notation,  \emph{linear loop programs} 
can be represented as
$$\textsf{while}\ (B x > b),\ \{x:=Ax+c\},$$ 
for $x$ and $c$ in $\R^n$, $b$ in $\R^m$, and $A$ and $B$ real matrices of 
size $n\times n$ and $m\times n$, respectively.  
%Without lost of generality, 
The termination analysis for this class of linear programs can be reduced to the termination problem of homogeneous programs with 
one loop condition, \emph{i.e.} when $m=1$, $b$ is zero and $c$ is a zero vector \cite{Braverman, tiwaricav04}.
The really difficult step being the reduction to $m=1$, while 
the reduction to $b$ and $c$ being zero is immediate. 
We focus on the termination of 
this type of program with one loop condition, and obtain results as sharp and complete as one could hope. 
At this point, it is worth mentioning some recent work on \emph{asymptotically non-terminating initial variable values} generation
techniques~\cite{TR-IC-14-03}.
Amongst many other results, we obtain methods that can be adapted here in order to extend our termination analysis for general linear programs, \emph{i.e.} when $m$ is arbitrary.
%Our resutls focus
%on termination for homogeneous linear programs as they could be extended to
%the class of affine linear programs using similar techniques developped in our
%recent work on asymptotically non-terminant initial variable values generation
%using the results of the present work, a complete and efficient termination analysis for general linear programs, \emph{i.e.} for arbitrary $m$. 

Despite tremendous progress over the years \cite{Braverman, Bradley05Manna, Chen2012,Cousot05-VMCAI, Cousot:2012, Cook:2006, Ben, 2013:POPL}, the problem of finding a practical, sound and complete method for determining termination or non termination remains very challenging for this class of programs, and for all initial 
variable values. 
We also note that some earlier works~\cite{TR-IC-13-07, TR-IC-13-08} have inspired the methods developed here.

We summarize our contributions as follows:
\vspace*{-1ex}
\begin{description}
\item[\it Preliminary result:]
%\begin{itemize}
%\item 
First we prove a sufficient condition for the termination of
homogeneous linear programs.  This result is also stated in~\cite{tiwaricav04}, but 
some shortcomings in that proof sketch require further elaboration. We
closed those gaps in a solid mathematical way, with some obstacles not being so
easy to overcome. 
We return to this point in more detail at
Remark~\ref{rem1}.  Our new proof of this sufficient condition
requires nontrivial topological and algebraic arguments.  On the other
hand, this sufficient condition is not a necessary condition for
termination of linear homogeneous programs.  Before we list our main
contributions, it is important to note that the
works~\cite{tiwaricav04,Braverman} produce some decidability results for
this type of programs.
However, for programs with one loop condition, our characterization of termination is 
much simpler, very explicit, and straightforwardly leads to 
much faster algorithm for checking termination. 
See also Section~\ref{discus} for a more detailed comparison.
% to Tiwari's and Braverman's methods~\cite{tiwaricav04,Braverman}.  %\qed 

%expertise in several independent mathematical fields. 
%We show how this sufficient condition can be used to determine termination of 
%linear programs. We also draw its limitations.      
\item[\it Main contributions:]\mbox{}\\
%\item 
\textit{(i)} %We then generalize the previous results. 
We give a \emph{necessary and sufficient condition} (NSC, for short) for the termination of linear homogeneous programs 
with one loop condition. In fact, this NSC exhibits a complete characterization of termination for such programs, and gives decidability results for all initial variable values.  
%\item 

\textit{(ii)} Moreover, departing from this NSC, we show the scalability of our approach by demonstrating that one can directly extract a sound and complete computational method to determine termination of such programs. 
We reduce the termination analysis to the problem of checking if a specific vector,
related to the loop encoding condition, belongs to a specific vector space related to the eigenvalues of the matrix encoding assignments to the loop variables.       
%\item 
%\noindent \textbf{$(iii)$} 
The analysis of our associated algorithms shows that our method has a much better computational time complexity. We show that the method, based on three computational steps running in polynomial time complexity, is of a lower complexity than basic routines that form the mathematical foundations of previous methods \cite{tiwaricav04,Braverman}. %FIXME 2: details on the comlexity.  
%\item 

\textit{(iii)} We provide theoretical results guaranteeing the soundness and completeness of the termination analysis while restricting  variable interpretations over a specific countable sub-ring of $\R^n$. In other words, we show that it is enough to interpret variable values over a specific countable 
field --- a number field, or even  its ring of integers, --- when one wants to check the termination over the reals. 
By so doing, we circumvent difficulties such as rounding errors. Those results enable our symbolic computational methods to rely on closed-form algebraic expression and numbers.%\qed
\end{description}
%\item We provide experiments associated to our prototype. The cpu timing results while determining termination or 
%non termination over a large number of random affine programs clearly demonstrate the practability of our approach. 
%\end{itemize}

%The reader can find all the complete proofs, written rigourously, in the annex of this article or in our associated 
%technical reports \cite{TR-IC-13-10} \cite{TR-IC-13-11}. 

The rest of this article is organized as follows. Section \ref{prelim} is a preliminary section where 
we introduce our computational model for programs, the notations for the rest of the paper, and some key notions of linear algebra used  to develop our computational methods. Section \ref{decid}  
develops %the main theoretical contributions of this work. %We present 
our theoretical results and a very useful necessary and sufficient condition, in Subsection \ref{newdecidability}, which allows us to propose the complete computational method illustrated in Section \ref{run-ex}, and fully described in Section \ref{algo}. %In Section \ref{P5P2toP1}, we show that our approach and 
%algortihms, scale in the handling of the class of affine programs. 
  In the important Section \ref{count}, we show that it is enough to interpret the variable values over a countable field in order to determine  program termination over the reals. 
We provide a discussion of related works in Section \ref{discus}.
%Finally, Section \ref{experiment} exhibits our experiments and 
Finally, Section \ref{conclusion} concludes the paper.

% !TeX spellcheck=en_US
\section{Linear Algebra and Linear Loop Programs}\label{prelim}

%Here, we define key notions of linear algebra that are central in the theoretical and algorithmic development of our methods. 
We recall classical facts from linear algebra.  Let $E$ be a real
vector space and $\mathbf{A}\in End_\R(E)$, the space of
$\R$-linear maps from $E$ to itself.  
Let $E^\star$ be the set of linear functionals in $E$.% the space of linear maps from $E$ to $\R$.
  We denote by
$\mathcal{M}(p,q,\R)$ the space of $p\times q$ matrices, and if $p=q$
we simply write $\mathcal{M}(p,\R)$. We will denote by $\mathbb{K}$
the $\R$ or $\C$ fields.  If $A\in\mathcal{M}(m,n,\mathbb{K})$, with entry $a_{i,j}$ in position
$(i,j)$, we will sometimes denote it by $(a_{i,j})$.  If $B$ is a
basis for $E$, we denote by $A_B=Mat_B(\mathbf{A})$ the matrix of
$\mathbf{A}$ in the basis $B$, and we have 
$A_B\in\mathcal{M}(n,\R)$.  Let $I_n$ be the identity matrix in
$\mathcal{M}(n,\R)$, and $\i$ the identity of $E$.  The transpose of
the matrix $A=(a_{i,j})$ is the matrix $A^{\top}=(b_{i,j})$ where
$b_{i,j} = a_{j,i}$.  The kernel of $A$, also called its
\emph{nullspace} and denoted by $Ker(A)$, is the set $\{v \in
\mathbb{K}^n\ |\ A\cdot v=0_{\mathbb{K}^m} \}$.  
%In fact, when we
%deal with square matrices, these Kernels are \emph{Eigenspaces}.  
Let
$A$ be a square matrix in $\mathcal{M}(n,\mathbb{K} )$.  A nonzero
vector $x \in \mathbb{K}$ is an eigenvector of $A$ associated with
an eigenvalue $\lambda\in \mathbb{K}$ if $A\cdot x = \lambda x$,
\emph{i.e.}, $(A-\lambda I_n)\cdot x = 0$.  The nullspace of
$(A-\lambda I_n)$ is called the \emph{eigenspace} of $A$ associated
with eigenvalue $\lambda$. A non-zero vector $x$ is said to be a
\emph{generalized eigenvector} of $A$ corresponding to $\lambda$ if
$(A-\lambda I_n)^k \cdot x = 0$ for some positive integer $k$. The
spaces $Ker((A-\lambda I_n)^k)$ form an increasing larger sequence of
subspaces of $E$, which is stationary for $k\geq e$, for some $e\leq
n$. We call the subspace $Ker((A-\lambda I_n)^e)= Ker((A-\lambda
I_n)^n)$ the \emph{generalized eigenspace} of $A$ associated with
$\l$, and its nonzero elements are exactly the generalized eigenvectors.
We denote by $\langle \ ,\ \rangle$ the canonical scalar product on
$\R^n$.
As it is standard in static program analysis, 
a primed symbol $x'$ refers to the next state value of $x$ after a
transition is taken.  Next, we present \emph{transition systems} as
representations of imperative programs and \emph{automata} as their
computational models.

\begin{df}
In a \emph{transition system} $\langle x, L, \mathcal{T}, l_0,\Theta \rangle$,   $x=(x_1, ...,x_n)$ is a set of variables, $L$ is a set of locations and $l_0\in L$ is the initial location. A \emph{state} is given by an interpretation of the variables in $x$. A \emph{transition} $\tau \in \mathcal{T}$ is given by a tuple $\langle l_{pre}, l_{post}, q_{\tau}, \rho_{\tau} \rangle$, where $l_{pre}$ and $l_{post}$ designate the pre- and post-locations of $\tau$, respectively, and the transition relation $\rho_{\tau}$ is a first-order assertion over $x\cup x'$. The transition guard $q_{\tau}$ is a conjunction of inequalities over $x$. $\Theta$ is the initial condition, given as a first-order assertion over $x$. The transition system is said to be \emph{linear} when $\rho_{\tau}$ is an affine form.\qed 
\end{df}

A loop program, defined next, is a special kind of transition system.
We also establish some matrix notations to represent loop programs, where the effects of 
\emph{sequential} linear assignments are described as \emph{simultaneous} updates. Departing from sequential instructions, we use syntatic and common propagation procedures to obtain the equivalent simultaneous systems expressed in matrix notations (see Definition \ref{mat}.).  

\begin{df}\label{mat} Let $P=\langle x, \{l\}, \mathcal{T}, l,\Theta \rangle$ be a transition system
with $x=(x_1,...,x_n)$ and $\mathcal{T}=\{\langle l, l, q_{\tau}, \rho_{\tau} \rangle\}$.  
Then $P$ is a \emph{linear loop program} if:
\begin{itemize}
\item The transition guard is a  conjunction of linear inequalities. We represent the loop condition in matrix form as $F x > b$ where $F\in\mathcal{M}(m,n,\R)$ and $b\in\R^m$.
%By $F x>b$ 
Which means that each coordinate of the vector $F x$ is  greater than the corresponding coordinate of vector $b$.
\item The transition relation is an affine or linear form. We represent the linear assignments in matrix form as $x:=Ax+c$, where $A\in\mathcal{M}(n,\R)$ and $c\in\R^n$.  \qed
\end{itemize}
%\item We denote by $P^{\H}$ the set of programs where all linear assignments consist of \emph{homogeneous} expressions, and where the linear loop condition  consists of at most one inequality. If $P$ is in $P^{\H}$, then $P$ will be interpreted in matrix terms as $\textsf{while}\ (w^{\top}x > 0), \ \{x:=Ax\}$, where $w$ is a $(n\times 1 )$-vector corresponding to the loop condition, and  $A\in \mathcal{M}(n,\R)$ is related to the list of assignments in the loop.
% We say that $P$ has a \emph{homogeneous} form and it will be denoted as $P(A,w)$.
\end{df}    
The most general \emph{linear loop program} $P=P(A,F,b,c)$ is thus written 
$$\textsf{while}\ (F x > b)\ \{x:=Ax+c\}.$$

In this work, one needs first to focus mainly on the following class of linear loop programs. 
\begin{df}
We denote by $P^{\H}$ the set of programs where all linear assignments consist of \emph{homogeneous} expressions, and where the linear loop condition  consists of at most one inequality.  \qed 
\end{df}

If $P$ is in $P^{\H}$, then $P$ will be interpreted in matrix terms as 
$$\textsf{while}\ (\langle f, x \rangle > 0) \ \{x:=Ax\},$$ 
where $f$ is a $(n\times 1 )$-vector corresponding to the loop condition, and  $A\in \mathcal{M}(n,\R)$ is related to the list of assignments in the loop. 
In this case, we say that $P$ has a \emph{homogeneous} form and it will be identified as $P(A,f)$.
    
Consider a program $P(A,f)$, where
$A\in\mathcal{M}(n,\R)$, $f\in\mathcal{M}(1,n,\R)$. 
Alternatively, we may consider $\A\in End_\R(E)$, $\f\in E^*$ and write  
$$P(\A,\f): \textsf{while}\ (\f(\xx) > 0)  \{\xx:=\A \xx\}.$$
Fixing a basis $B$ of $E$ we can write $A=Mat_B(\A)$, $f=Mat_B(\f)$, $x=Mat_B(\xx)$, and so on. 
We now define  termination for such programs. 

\begin{df}\label{ter}
Program $P(\A,\f)$ terminates on input $\xx\in E$ if and only if
there exists $k \geq 0$ such that $\f(\A^{k}(\xx))$ is not positive.
Alternatively, for $A\in\mathcal{M}_n(\R)$, and $f\in\mathcal{M}_{1,n}(\R)$, we
say that $P(A,f)$ terminates on input $x\in \R^n$, if and only if
there exists $k\geq 0$, such that $\langle A^{k} x , f\rangle$ is not
positive.   
Thus, a program $P(\A,\f)$ is non-terminating if and only if there
exists an input $\xx\in E$ such that $\f(\A^{k}(\xx)) > 0$ for all
$k\geq 0$. In matrix terms, $P(A,f)$ is non-terminating on input
 $x\in \R^n$ if and only if $\langle A^k x, f \rangle > 0$ for
all $k\geq 0$.
\qed
\end{df}

% !TeX spellcheck=en_US
\section{Linear Program Termination}\label{decid}

%In this section we introduce the theoretical foundations of our approach.  Here, we provide decidability results for the termination of the class of linear programs. 

%For this section, and without lost of generality, it is enough to consider only the class of homogeneous linear programs $P^{\H}$ (see Definition \ref{classprog}).  In fact, the problem of termination of linear programs in $P^{\mathbb{A}}$ (i.e. the class of affine programs, see Definition \ref{classprog}) reduces to the problem of termination of homogeneous linear programs $P^{\H}$ \cite{tiwaricav04}.  

First we prove a \emph{sufficient} condition for the termination of homogeneous linear programs, already stated in \cite{tiwaricav04}.
We note that he proof of sufficiency in \cite{tiwaricav04} does not go through, and needed to be amended,
which was not a trivial task.
Then we present the main result, which provides the first \textit{necessary and sufficient} condition for the termination problem for the class of linear homogeneous programs. %Those decidability results lead us to a complete method, associated to fast algorithms to determines termination of linear programs.    

\subsection{Sufficiency and Homogeneous Linear Programs}\label{sufficient}

We prove a sufficient condition for the termination of programs $P(A,f)\in P^{\H}$, written 
$$\textsf{while}\ (f^{\top}x >0)\ \{x:=Ax\}.$$ 

\begin{thm}\label{resprinc}
Let $n$ be a positive integer, and let $P(A,f)\in P^{\H}$. %defined by the linear assignements encoded by a matrix $A$ in $\mathcal{M}(n,\R)$, and the inequality loop condition described by the vector $w\in \R^n-\{0\}$. 
If $P(A,f)$ is non-terminating,
%, (i.e. if there exists a vector $x\in \R^n$ such that $\langle A^k x,f \rangle > 0$ for all $k\geq 0$, see Definition \ref{ter}), 
then $A$ has a positive eigenvalue.\qed
\end{thm}

%=======PROOF==========

\textit{In the following discussion, we provide the complete proof of Theorem \ref{resprinc}}.
Before we complete the proof, which is a mix of topological and algebraic arguments, we need first to state the following lemmas and propositions.   %The reader can find the complete proof in the annex of this article or in our associated technical report \cite{TR-IC-13-10}. 
We first recall some basic facts about generalized eigenspaces. 
Let $E$ be an $\R$-vector space of finite dimension, and let $\A\in End_\R(E)$.
%, the space of linear maps from $E$ to itself. 
Let $E'$ be a subspace of $E$.
We say that $E'$ is $\A$-stable if $\A(E')\subseteq E'$. If $\l\in \R$, we denote by $E_\l(\A)$ the subspace $\{x\in E | \exists k\geq 0, (\A-\l\i)^k(x)=0\}$. This space is non zero  if and only if the input vector $x$ is an eigenvector of $\A$.
In this case, it is called the generalized eigenspace corresponding to $\l$. If 
$\chi_{\A}$ is the characteristic polynomial of $\A$, if  $d_{\l}$ is the multiplicity of the monomial $(X-\l)$ in $\chi_{\A}(X)$, which may be $0$ if $\l$ is not an eigenvalue,
 then $E_\l(\A)=Ker(\A-\l \i)^{d_{\l}}$. It is obvious that $E_\l(\A)$ is $\A$-stable. We denote by $Spec(\A)$ the set of real eigenvalues of $\A$. The following property of generalized eigenspaces was stated in the preliminaries.

\begin{prop}\label{maj}
Let $E$ be an $\R$-vector space of finite dimension, and let $\A$ belong to $End_\R(E)$. Then $E_\l(\A)=Ker(\A-\l \i)^{d_{\l}}$, for some $d_{\l}\leq n$.
In particular, 
$E_\l(\A)=Ker(\A-\l \i)^n$.\qed
\end{prop}

\begin{proof}
 We just said that one can choose $d_{\l}$ to be such that $(X-\l)^{d_{\l}} \backslash \chi_{\A}$.
 Hence, $d_{\l}\leq d^{\circ}(\chi_{\A})=n$ (with $d^{\circ}$ beeing the standart notation for polynomial degree.).%\qed
\end{proof}

%\begin{proof}
% We just said that one can choose $d$ to be such that %$(X-\l)^d\backslash \chi_u$, hence $d\leq d^{\circ}(\chi_u)=n$.\qed
%\end{proof}

We will also need the following lemma.

%\begin{LM}\label{standard}
%In the previous situation, there is a supplementary space $E'$ of $E_\l(\A)$ (i.e. $E=E_\l(\A)\oplus E'$), and two polynomials $P$ and $Q$ in $\R[X]$, 
%such that $P(\A)$ is the projection on $E_\l(\A)$ with respect to $E'$, and $Q(\A)$ is the projection on $E'$ with respect to $E_\l(\A)$. 
%In particular $E'$ is also $\A$-stable, and for any 
%$\A$-stable subspace $L$ of $E$, we have $L=L\cap E_\l(\A) \oplus L\cap E'$.  \qed
%\end{LM}

%\begin{proof}
%Let $\chi_u=(X-\l)^dQ$, with $Q(\l)\neq 0$. By the Kernel's decomposition Lemma, we have $$E=Ker(u-\l I_d)^d\oplus Ker(Q(u)).$$ We set $E'=Ker(Q(u))$. 
%It is thus $u$-stable. Moreover, by Bezou's identity, there are $P$ and $P'$ in $\R[X]$, such that 
%$$P(u)\circ (u-\l I_d)^d+P'(u)\circ Q(u)=I_d,$$ then we set $B=P(X-\l)^d$, and $A=P'(u)\circ Q(u)$. 

%Finally, if $L$ is $u$-stable, we always have $$L\cap E_\l(u) \oplus L\cap E'\subset L.$$ Now write an element 
%$l$ of $L$ as $l_1+l_2$, with $l_1\in E_\l(u)$, and 
%$l_2\in E'$, we have $A(u)(l)=l_1$, but $L$ being $u$-stable, it is $A(u)$-stable as well, hence $l_1\in L$, similarly we have $l_2\in L$, thus 
% $$L=L\cap E_\l(u) \oplus L\cap E'.$$ \qed
%\end{proof}

\begin{LM}\label{linearforms}
Let $E^*$ be the space $Hom_\R(E,\R)$, where $E$ is a finite dimensional vector space, and $f_0,\dots, f_m$ be linear forms in $E^*$. 
Then this family spans $E^*$ 
if and only if $\cap_{i=0}^m Ker(f_i)=\{0\}$.\qed
\end{LM}

\begin{proof}[Proof of Lemma \ref{linearforms}]
In the following we use the notation $Vect(v_1,...,v_u)$ to describe the vector space spaned by the elements $v_1,...,v_u$.
Suppose that $f_0,\dots, f_m$ spans $E^*$.
If $x$ belongs to $\cap_{i=0}^m Ker(f_i)$, then $x$ belongs to the kernel of any element of $E^*$. 
But then, if $B=(e_1,\dots,e_n)$ is a basis of $E$, and $B^*=(e_1^*,\dots,e_n^*)$ is its dual basis, we have $x=x_1.e_1+\dots+x_n.e_n$, and $e_i^*(x)=x_i=0$.
Hence, $x=0$. 
Conversely, if $\cap_{i=0}^m Ker(f_i)=\{0\}$, let $g_1,\dots,g_r$ be a maximal linearly independent family 
in $f_0,\dots, f_m$. 
Hence, $Vect(g_1,\dots,g_r)=Vect(f_0,\dots, f_m)$.  
We thus have $r\leq n$ because $dim(E^*)=dim(E)=n$ and $\cap_{i=1}^r Ker(g_i)=\{0\}$. 
If $r$ was strictly smaller than $n$, then $\cap_{i=1}^r Ker(g_i)$ would be an intersection of $r$ subspaces of co-dimension $1$.
Hence, it would be of co-dimension at most $r$, \emph{i.e.}, $\cap_{i=1}^r Ker(g_i)$ would be of dimension at least $n-r>0$, which is 
a contradiction.
Thus $r=n$, and $(g_1,\dots,g_r)$ is a basis of $E^*$.
It follows that $Vect(f_0,\dots, f_m)=E^*$.%\qed
\end{proof}

%\begin{proof}
%Suppose that $f_0,\dots, f_m$ spans $E^*$, then if $x$ belongs to $\cap_{i=0}^m Ker(f_i)$, then $x$ belongs to the kernel of any element of $E^*$. 
%But then, if $B=(e_1,\dots,e_n)$ is a basis of $E$, and $B^*=(e_1^*,\dots,e_n^*)$ is its dual basis, we have $x=x_1.e_1+\dots+x_n.e_n$, and $e_i^*(x)=x_i=0$, hence $x=0$. 

%Conversely, if $\cap_{i=0}^m Ker(f_i)=\{0\}$, Let $g_1,\dots,g_r$ be a maximal linearly independent family in 
%$f_0,\dots, f_m$, hence $$Vect(g_1,\dots,g_r)=Vect(f_0,\dots, f_m).$$ We thus have $r\leq n$ (because $dim(E^*)=dim(E)=n$), and $\cap_{i=1}^r Ker(g_i)=\{0\}$. 
%If $r$ was strictly smaller than $n$, then $\cap_{i=1}^r Ker(g_i)$ would be an intersection of $r$ subspaces of codimension $1$, hence it would be of co-dimension at most $r$, i.e. $\cap_{i=1}^r Ker(g_i)$ would be of dimension at least $n-r>0$, which is absurd, thus $r=n$, and $(g_1,\dots,g_r)$ is a basis of $E^*$, thus $$Vect(f_0,\dots, f_m)=E^*.$$\qed
%\end{proof}

Before proving Lemma \ref{quotient}, we recall and prove the following standard lemma.

\begin{LM}\label{standard}
Let $\A$ be an endomorphism of a real vector space $E$, and let $\l$ be an eigenvalue of $\A$. There is a supplementary space $E'$ of $E_\l(\A)$, \emph{i.e.}, $E=E_\l(\A)\oplus E'$), and two polynomials $C$ and $D$ in $\R[X]$, 
such that $C(\A)$ is the projection on $E_\l(\A)$ with respect to $E'$, and $D(\A)$ is the projection on $E'$ with respect to $E_\l(\A)$. 
In particular $E'$ is also $\A$-stable, and for any 
$\A$-stable subspace $L$ of $E$, we have $L=L\cap E_\l(\A) \oplus L\cap E'$.  \qed
\end{LM}
\begin{proof}
Let $\chi_{\A}=(X-\l)^dQ$, with $Q(\l)\neq 0$. By the kernel decomposition lemma, we have $$E=Ker(\A-\l I_d)^d\oplus Ker(Q(\A)).$$ We set $E'=Ker(Q(\A))$. 
It is thus $\A$-stable. Moreover, by Bezout's identity, there are $P$ and $P'$ in $\R[X]$, such that 
$$P(u)\circ (\A-\l I_d)^d+P'(u)\circ Q(\A)=I_d.$$ 
We set $C=P(X-\l)^d$, and $D=P'(\A)\circ Q(u)$. 
Finally, if $L$ is $\A$-stable, we always have $$L\cap E_\l(\A) \oplus L\cap E'\subset L.$$ Now write an element 
$l$ of $L$ as $l_1+l_2$, with $l_1\in E_\l(\A)$, and 
$l_2\in E'$.
We get $B(\A)(l)=l_1$. 
But $L$ being $\A$-stable, it is also $D(\A)$-stable as well.
Hence, $l_1\in L$. Similarly we have $l_2\in L$, thus 
 $$L=L\cap E_\l(\A) \oplus L\cap E',$$ %\qed
 completing the proof.
\end{proof}

We will use the following result about quotient vector spaces.

\begin{LM}\label{quotient}
 Let $E$ be an $\R$-vector space, let $\A\in End_\R(E)$, and suppose that $L$ is a $\A$-stable subspace of $E$. 
Let $\overline{\A}:E/L\rightarrow E/L$ be the element of $End_\R(E/L)$ defined by $\overline{\A}(x+L)=\overline{\A}(x)+L$.
Then $Spec(\overline{\A})\subset Spec(\A)$. More generally, for any $\l\in Spec(\overline{\A})$, the generalized eigenspace $E_\l(\A)$ maps surjectively to $E_\l(\overline{\A})$ 
in $E/L$. \qed
\end{LM}

\begin{proof}[Proof of Lemma \ref{quotient}]
Let $B_1$ be a basis for $L$, and $B_2$ be a basis for any supplementary space. Call $\overline{B_2}$ the image of the elements of $B_2$ in 
$\overline{E}=E/L$.
Then $\overline{B_2}$ is a basis of $\overline{E}$. 
With $B=B_1\cup B_2$,
% it is a basis of $B$, and 
$Mat_B({\A})$ is of the form $$\begin{pmatrix} X & Y \\ 0 & Z \end{pmatrix}.$$ 
Then $X=Mat_{B_1}({\A}_{|L})$,  $Z=Mat_{\overline{B_2}}(\overline{{\A}})$, 
and the second statement follows from this second fact.\\
Now if $\overline{x}$ belongs to $E_\l(\overline{{\A}})$, then 
$(\overline{{\A}}-\l\overline{I_d})^a\overline{x}=\overline{0}$ for some $a\geq 0$. This means that $({\A}-\l I_d)^a x\in L$. 

We write $x=x_\l+x'\in E_\l({\A}) \oplus E'$, for $E'$ as in Lemma \ref{standard}. 
Then  $({\A}-\l I_d)^a x= ({\A}-\l I_d)^a x_\l + ({\A}-\l I_d)^a x'$, with
$ ({\A}-\l I_d)^a x_\l\in E_\l({\A})$, and 
$({\A}-\l I_d)^a x'\in E'$. Let $d$ be the  multiplicity of $\l$ as a root of $\chi_{\A}$.
For $k$ large enough such that $kd\geq a$, 
we have $({\A}-\l I_d)^{kd} x_\l=0$ and 
$({\A}-\l I_d)^{kd} x=({\A}-\l I_d)^{kd} x'$. 
Taking $P\in \R[X]$ as in the proof of Lemma \ref{standard}, we have that 
$P({\A})\circ ({\A}-\l I_d)^{d}$ is the identity when restricted to $E'$.
In particular, this implies that 
$$x'=P({\A})^k({\A}-\l I_d)^{kd} x,$$ and thus $x'\in L$. 
Finally, we obtain $\overline{x}=\overline{x_\l}$, and this concludes the proof, as 
$x_\l \in  E_\l({\A})$. %\qed
\end{proof}

%\begin{proof}
%Let $B_1$ be a basis of $L$, and $B_2$ be a basis of any supplementary space. Call $\overline{B_2}$ the image of the elements of $B_2$ in 
%$\overline{E}=E/L$, then $\overline{B_2}$ is a basis of $\overline{E}$. Let $B=B_1\cup B_2$, it is a basis of $B$, and 
%$Mat_B(u)$ is of the form $$\begin{pmatrix} X & Y \\ 0 & Z \end{pmatrix}.$$ Then $X=Mat_{B_1}(u_{|L})$, and $Z=Mat_{\overline{B_2}}(\overline{u})$, 
%and the second statement follows from this second fact.\\
% Now if $\overline{x}$ belongs to $E_\l(\overline{u})$, then 
%$$(\overline{u}-\l\overline{I_d})^a\overline{x}=\overline{0}$$ for some $a\geq 0$. This means that $(u-\l I_d)^a x\in L$. 

%We write $x=x_\l+x'\in E_\l(u) \oplus >E'$, for $E'$ as in the Lemma \ref{standard}. Then  $(u-\l I_d)^a x= (u-\l I_d)^a x_\l + (u-\l I_d)^a x'$, and 
%$ (u-\l I_d)^a x_\l\in E_\l(u)$, and 
%$(u-\l I_d)^a x'\in E'$. Let $d$ be $\l$'s multiplicity as a root of $\chi_u$, for $k$ large enough $kd$ such that $kd\geq a$, 
%we have $(u-\l I_d)^{kd} x_\l=0$ and 
%$(u-\l I_d)^{kd} x=(u-\l I_d)^{kd} x'$. But take $P\in \R[X]$ as in the proof Lemma \ref{standard}, we obtain that 
%$P(u)\circ (u-\l I_d)^{d}$ is the identity when restricted to $E'$, in particular, this implies that 
%$$x'=P(u)^k(u-\l I_d)^{kd} x,$$ and thus $x'\in L$. Finally, we obtain $$\overline{x}=\overline{x_\l},$$ and this ends the proof as 
%$x_\l \in  E_\l(u)$. \qed
%\end{proof}

We say that a subset of $\R^n$ is a convex cone if it is convex, and it is also stable under multiplication by elements of $\R_{>0}$. It is obvious that an intersection of convex cones is still a convex cone,
and so one can speak of the convex cone spanned by a subset of $\R^n$.

\begin{prop}\label{casregulier}
Let $C$ be a convex cone of $\R^n$.
Assume that $C$ is non reducible to zero, and is contained in the closed cone 
$$\D=\{x=(x_1, ..., x_n)\in \R^n\,|\, \forall \ i , x_i\geq 0\}.$$ If $A$ is an invertible endomorphism of $\R^n$, with $A(C)\subset C$, then $A$ has a positive eigenvalue. \qed 
\end{prop}

\begin{proof}
Consider $C'=C-\{0\}$.
Then $C'$ is also a convex cone. It is obviously still stable under multiplication by elements of $\R_{>0}$. Moreover, 
if $x$ and $y$ belong to $C'$, then  the vector $tx+(1-t)y$ belongs to $C$ by convexity, for $t\in[0,1]$.
But it cannot be equal to zero, as both $x$ and $y$ have non negative coefficients, this would imply that $x$ or $y$ is null, which is a contradiction.

Now let $H_1$ be the affine hyperplane $H_1=\{x \in \R^n, x_1+\dots x_n=1\}$, and let $f$ be the linear form on $\R^n$ 
defined by $f:x\mapsto x_1+\dots + x_n$, so that $H=f^{-1}(\{1\})$. This linear form is positive on $\D$, and so
we can define the projection $p:\D-\{0\}\rightarrow H$ given by 
$$x\mapsto \frac{1}{f(x)}x.$$ i
It is obviously continuous. We call $C_1$ the set $p(C')$.
We claim that $C_1=C'\cap H_1$ and, in particular, it is convex. 
Indeed, $C_1\subset H_1$ by definition, and $C_1\subset C'$ because $C'$ is stable under $\R_{>0}$. Conversely, the restriction of $p$ to $C'\cap H_1$ is the identity, and so $C_1$ contains $C'\cap H_1=p(C'\cap H_1)$. It is also clearly stable under the continuous map 
$$s=p\circ A:\D-\{0\}\rightarrow H_1,$$ as $A(C')\subset C'$. In particular, its closure $\overline{C_1}$ is stable under $s$ as well. 
It is convex and compact, as a closed subset 
of the compact set $$\{x \in \R^n, \forall \ i, \ x_i\geq 0,\ x_1+\dots +x_n=1\}.$$ 
According to Brouwer's fixed point theorem, this implies that $s$ has a fixed $x$ point in $\overline{C_1}\subset \D-\{0\}$.
But we then have $A(x)=f(x)x$. As $f(x)>0$ for any $x$ in $\D-\{0\}$.
This proves the lemma.%\qed
\end{proof}

Finally we will prove the following statement equivalent to Theorem \ref{resprinc}.
We just rewrite the statement of Theorem \ref{resprinc} in terms of morphisms, which are
more convenient to work with.

\begin{thm}\label{casgen}
Let $E$ be an $\R$-vector space of dimension $n$, let $\A$ be a endomorphism of $E$, and let $f$ be a nonzero linear form on $E$. 
If there exists a vector $x\in E$ such that $\f({\A}^k(x))>0$ for all $k\geq 0$, then $\A$ has a positive eigenvalue. \qed
\end{thm}

\begin{proof}
We prove the result by induction on $n$. 
When $n=1$, we can identify $E$ with $\R$. Then $\A$ is of the form $x\mapsto t_{\A}.x$, for some nonzero $t_{\A}$, and $\{\f> 0\}$ 
is either $\R_{> 0}$, or $\R_{< 0}$. Hence, $x$ belongs to $\R_{> 0}$, or to $\R_{<0}$, and $t_{\A}^k.x$ belongs to the same half-space for every $k\geq 0$.
Hence, $t_{\A}>0$. 

Now if $\A$ is non invertible, we can replace $E$ by the image of $\A$, $Im(\A)$, and $x$ by $\A(x)$, so that the hypothesis are still verified by $\A$'s restriction to $Im(\A)$. But since $Im(\A)$ is a subspace of $E$ of strictly smaller dimension, we get the result 
using the induction hypothesis.
We are thus left with the case when $\A$ is invertible. 
Let $m$ be the maximal non negative integer such that $(\f,\f\circ \A,\dots, \f\circ {\A}^m)$ is a linearly independent family of $E^*$. It is easy to see that $L=\cap_{k\geq 0} Ker(\f\circ {\A}^k)$ is equal to $\cap_{k=0}^m Ker(\f\circ {\A}^k)$.
Hence, it is $\A$-stable. The space $L$ is a proper subspace of $E$ because it is contained in $Ker(\f)$.
Taking the quotient space $\overline{E}=E/L$, the linear map $\A$ induces $\overline{\A}:\overline{E}\rightarrow \overline{E}$, and $\f$ induces a linear form $\overline{\f}$ on $\overline{E}$. By letting $\bar{x}$ be the image of $x$ in $E$, the quadruplet $(\overline{E},\overline{\A},\overline{\f},\bar{x})$ still satisfies the hypothesis of the theorem. 
If $L$ is not zero,  using the induction we conclude that the linear map $\overline{\A}$ has a positive eigenvalue $\l>0$.
But $\l$ is necessarily an eigenvalue of $\A$ by Lemma \ref{quotient}, 
and we are done in this case. 
Finally, assume that $L=\{0\}$.
Then $(e_1^*=\f,e_2^*=\f\circ \A,\dots, e_n^*=\f\circ {\A}^m)$ is a basis of $E^*$. 
In particular $m=n-1$, according to Lemma \ref{linearforms}. 
Take $(e_1,\dots,e_n)$ as its dual basis in $E$, and identify $E$ with $\R^n$, given this basis. Then ${\A}^k(x)$ belongs to the space $\{v\,|\,\forall i, v_i>0\}\subset \D$ 
for all $k\geq 0$.
Hence, the convex cone 
$C$ is spanned by this family as well. It is clearly $\A$-stable, and is not reduced to zero as it contains $x$. We conclude by applying Proposition \ref{casregulier}.%\qed
\end{proof}
%==========PROOF============

This also concludes the proof of Theorem \ref{resprinc}, as Theorem \ref{casgen} is an equivalent statement written in terms of the morphisms $A=Mat_B(\A)$) and $f=Mat_B(\f)$. 
%Theorem \ref{resprinc} provides a sufficient condition for the termination of linear program. In other words, 
Theorem \ref{resprinc} says that the linear program terminates when $A$ has  no positive eigenvalue.
But one cannot conclude on the termination problem using Theorem \ref{resprinc} 
when $A$ has at least one positive eigenvalue. As we already mentioned, Theorem \ref{resprinc} is stated in~\cite{tiwaricav04}.
But the proof given therein contains certain flaws that we now expose. 
%of the result contains non trivially mistakes, 
\begin{rem}\label{rem1}
The argument of \cite{tiwaricav04} applies the Brouwer's fixed point theorem to a subspace of the projective space $P(\R^n)$, and not $\R^{n-1}$ as stated in \cite{tiwaricav04}. However, this is not an Euclidian space, and so convexity is not well defined in it.
Hence, one cannot apply Brouwer's fixed point theorem to such a set. Moreover, using notation as in the proof of Theorem 1 in \cite{tiwaricav04}, the closure $NT'$ of the set $NT$ can contain zero.
For example as soon as all, real or complex, eigenvalues of $A$ have their module less than $1$. Hence, its image in $P(\R^n)$ is not well defined. The case of $NT'$ containing zero 
raises a serious problem that needs to be treated carefully.
We circumvent it by taking quotients by $L$ in our proof.\qed
\end{rem}

Theorem \ref{resprinc} provides a sufficient condition for the termination of linear programs. In other words, Theorem \ref{resprinc} says that the linear program terminates when there is no positive eigenvalues.
But one can not conclude on the termination problem using Theorem \ref{resprinc} if there exists at least one positive eigenvalue. Intuitively, we could say that Theorem \ref{resprinc} provides us with a decidability result for the termination problem considering the subclass of linear program where the associated assignment matrix $A$ has no positive eigenvalues.
% (i.e., all eigenvalues are complex or negative). 
In the following example, we illustrate situations where Theorem \ref{resprinc} applies and when it does not. 

\begin{figure}
\begin{subfigure}[b]{.5\linewidth}
\centering
\begin{lstlisting}
/*...*/
 while(3x - y > 0){
   x := 3x - 2y;
   y := 4/3x - 5/3y;
   }
/*...*/
\end{lstlisting}
\caption{}\label{fig:1a}
\end{subfigure}%
\begin{subfigure}[b]{.5\linewidth}
\centering
\begin{lstlisting}
/*...*/
 while(z > 0){
   x:= x + y;
   z:= -z;
   }
/*...*/
\end{lstlisting}
\caption{}\label{fig:1b}
\end{subfigure}
\caption{Examples of homogeneous linear programs}\label{fig:1}
\end{figure}

\begin{ex}\label{ex-sufficient}
Consider the homogeneous linear program \ref{fig:1a} denoted by $P(A,v)$, and
depicted in Figure \ref{fig:1}. 
The associated matrix $A=\begin{pmatrix} 3 & -2 \\ 4 & -1\end{pmatrix}$ correspond to the \emph{simultaneous updates} representing the \emph{sequential loop assignments}, and the vector $v$ encoding the loop condition, is   $v=(3 , -1)^{\top}$.  The eigenvalues of  $A$ are the complex numbers: $1+2i$ and $1-2i$. As $S$ does not have any positive eigenvalues, we can use Theorem \ref{sufficient} and conclude that  program $P(A,v)$ terminates on all possible inputs.\qed
\end{ex}

\begin{ex}\label{ex-limit-suff}
Now consider the homogeneous linear program \ref{fig:1b} depicted in Figure \ref{fig:1}, denoted by $P(A_1,v_1)$. The associated matrix $A_1$ representing the \emph{simultaneous updates} 
is given by $A_1=\begin{pmatrix} 1 & 1 & 0\\ 0 & 1 & 0\\ 0 & 0 & -1\end{pmatrix}.$
Its eigenvalues are $1$ and $-1$. As $A$ has a positive eigenvalues, one can not determine the termination  of $P(A_1,v_1)$ using Theorem \ref{sufficient}.  
In the following sections we will see how to handle this case in an automated and efficient.\qed  
\end{ex}    

In the next subsection we generalize Theorem \ref{sufficient}, obtaining stronger results.

\subsection{Necessity and Sufficiency for Termination of Linear Programs}\label{newdecidability}

%In this section, we provide a in order to obtain a complete decidability result leading us to a sound and complete methods with very few computational steps executed by fast algorithms. 

Theorem \ref{rachidiantheory}  provides a necessary and sufficient condition for the termination of programs $P(A,f)\in P^{\H}$ 
$$\textsf{while}\ ( f^{\top}x  >0)\ \{x:=Ax\}.$$ %This result characterizes the termination in terms of the generalized eigenvalues of $v$.  

\begin{thm}\label{rachidiantheory}
Let $A\in \mathcal{M}_n(\R)$ and let $f\neq 0$ be in $\R^n$. Then program $P(A,f)$
$$\textsf{while}\ ( f^{\top} x  >0)\ \{x:=Ax\}$$ 
terminates if and only if for every positive eigenvalue $\l$ of $A$, the generalized eigenspace $E_\l(A)$ is orthogonal to $f$, \emph{i.e.}, $f^{\top} E_\l(A)=\langle f,E_\l(A)\rangle =0$.\qed
\end{thm}
%======PROOF
%\textbf{Here we provide the complete proof of Theorem \ref{rachidiantheory}}. 
In order to prove Theorem \ref{rachidiantheory} we first restate it  in equivalent linear algebraic terms.

\begin{thm}\label{endo-thm}
Let $E$ be an $\R$-vector space of finite dimension $n$, let $\A$ be an endomorphism of $E$, and let $\f$ be a nonzero linear form on $E$. 
Then there exists a vector $x\in E$ with $\f({\A}^k(x))>0$ for all $k\geq 0$ if and only if there is $\l>0$ in $Spec(\A)$ 
such that $E_\l(\A)\not\subset Ker(\f)$.\qed
\end{thm}

\begin{proof}
 First suppose that there is a $\l>0$ in  $Spec(\A)$ with $E_\l(\A)\not\subset Ker(\f)$. Then there is some $r\geq 1$ such that 
$Ker(\A-\l \i)^{r-1}\subset Ker(\f)$.
But we also have  
$Ker(\A-\l \i)^{r}\not\subset Ker(\f)$. 
Let  $x$ be an element of $Ker(\A-\l \i)^{r}-Ker(\f)$  such that $\f(x)>0$.
This is always possible because $Ker(\A-\l \i)^{r}-Ker(\f)$ is stable under $y\mapsto -y$. Because $x\in Ker(\A-\l \i)^r$, it is clear that 
$\A(x )-\l x\in Ker(\A-\l \i)^{r-1}$. 
Let $L$ be $Ker(\A-\l \i)^{r-1}$, and let $\overline{E}=E/L$. As $L$ is $\A$-stable,  $\overline{\A}$ is well defined, and $\overline{\A}(\overline{x})=\l\overline{x}$ because 
$\A(x)-\l x\in L$. Moreover, $L\subset Ker(\f)$.
Hence, $\overline{\f}$ is well defined and $\overline{\f}(\overline{\A}^k(\overline{x}))=\f({\A}^k(x))$ for every $k\geq 0$. As $\overline{\A}^k(\overline{x})=\l^k \overline{x}$, we deduce that $\f({\A}^k(x))=\l^kf(x)>0$ for all $k\geq 0$.\\
Conversely, suppose that there exists a vector $x\in E$, such that $\f({\A}^k(x))>0$ for all $k\geq 0$.
We prove by induction on $n$ that $\A$ has an eigenvalue $\l>0$ such that $E_\l(\A)$ is not contained in $Ker(\f)$. 
If $n=1$, then $\A: t \mapsto \l t$ for $\l\in \R$, and so, $\l^k(\f(x))>0$ for all $k\geq 0$.
This implies $\l>0$, and so $E_\l(\A)=E$ is not be contained in $Ker(\f)$. 
If $n>1$, according to Theorem \ref{casgen} we know that $\A$ admits a positive eigenvalue $\mu$. 
If $E_\mu(\A)$ is not a subset of $Ker(\f)$ we are done. If $L=E_\mu(\A)\subset Ker(\f)$, we consider $\overline{E}=E/L$. This vector space is of dimension less than $n$ and so
$\overline{\f}(\overline{\A}^k(\overline{x}))=\f({\A}^k(x))>0$ for all $k\geq 0$. 
By the induction hypothesis, there is some $\l>0$ in $Spec(\overline{\A})$ such that $E_\l(\overline{\A})\not\subset Ker(\overline{\f})$. But $\l$ belongs to $Spec(\A)$ according to Lemma \ref{quotient}, and $E_\l(\A)$ maps surjectively on $E_\l(\overline{\A})$ according to this same Lemma. 
In particular, we have $\overline{\f}(E_\l(\overline{\A}))=\f(E_\l(\A))$, but the 
left hand side is not reduced to zero in this equality.
Hence, 
$\f(E_\l(\A))\neq \{0\}$, \emph{i.e.}, $E_\l(\A)\not\subset Ker(\f)$, concluding the proof. %\qed
\end{proof}

This argument proves  Theorem \ref{rachidiantheory} as it is a direct corollary of Theorem \ref{endo-thm} with $A=Mat_B(\A)$ and $f=Mat_B(\f)$. 
%\begin{cor}
%Let $A\in \mathcal{M}_n(\R)$ and $v\neq 0\in \R^n$. The program $P_1:\{x:=Ax,<v,x>>0\}$ terminates if and only if for every positive eigenvalue $\l$ of $A$,
%the generalised eigenspace $E_\l(A)$ is orthogonal to $v$ (i.e. $<E_\l(A),v>=0$).\qed
%\end{cor}
%========END============
%FIXME: next para graph could be removed.
Theorem \ref{rachidiantheory} gives a necessary and sufficient condition that we can use as the foundation to build a complete procedure for checking termination.  In order to determine termination, we have to check, for each positive eigenvalue, if the vector $f$, encoding the loop condition, is orthogonal to the associated generalized eigenspace. In other words we want to verify if $f$ is orthogonal to the nullspace $Ker((A-\l I_n)^n)$. 

\begin{ex}\label{ex-necessary-sufficient}
Consider the program \ref{fig:1b} depicted in Figure \ref{fig:1} that we denoted as $P(A_1,v_1)$. The matrix $A_1$ is given in Example \ref{ex-sufficient}. The vector encoding the loop condition is $v_1=e_3=(0 , 0 , 1)^{\top}$. We recall that $A_1$ has eigenvalues $1$ and $-1$. The generalized eigenspace $E_1(A_1)$ is equal to $Vect (e_1,e_2)$, where $e_1$ and $e_2$ are the first two vectors of the canonical basis of $\R^3$. Hence $E_1(A_1)$ is orthogonal 
to $v_1$. According to Theorem \ref{rachidiantheory}, program $P(A,w)$ terminates.\qed 
\end{ex}
\begin{ex}
Now we change the loop condition of  program \ref{fig:1b}, depicted in Figure \ref{fig:1}, to  $(y>0)$. Then, we obtain the program $P(A_1,v_2)$ with the new considered loop condition encoded as  $v_2=e_2=(0 , 1 , 0)^{\top}$. The eigenvalues of $A_1$ are (still) $1$ and $-1$, and the generalized eigenspace $E_1(A_1)=Vect (e_1,e_2)$. Hence $E_1(A)$ is not orthogonal to $v_2$, because it contains $v_2$. Theorem \ref{rachidiantheory} tells us the program $P(A_1,v_2)$ does not terminate.\qed 
\end{ex}
In both of these examples, we are able to determine the termination or nontermination 
of the corresponding program using Theorem \ref{rachidiantheory}.  On the other hand, 
Theorem \ref{resprinc} does not allow us to conclude 
anything about the termination of these programs, since the assignment matrix $A'$ exhibit at least one positive eigenvalue.  
In order to avoid the computation of  basis for generalized eigenspaces, we first introduce the space $Row\_Space(M)$, and use the next lemma. 
If $M\in \mathcal{M}(m,n,\R)$, then  $Row\_Space(M)$ denotes the vector subspace of $\R^n$ spanned by the row vectors of $M$.

\begin{LM}\label{lm1}
Let $M$ be a matrix in $\mathcal{M}(m,n,\mathbb{K})$. Then every vector in the nullspace of $M$ is orthogonal to every vector in  $Row\_Space(M)$.\qed 
\end{LM}     

\begin{proof}
Let $w\in Ker(M)$, and let $v$ be in the column space of $M^{\top}$. We denote by $\{c_1, ..., c_m\}$ the set of column vectors of $M^{\top}$. Then, there exists a vector $k\in \R^m$ such that $v=M^{\top}\cdot k$,
since $v$ is a linear combination of the column vectors of $M^{\top}$. 
Now we have $$<w,v>=w^{\top}\cdot v = w^{\top}\cdot M^{\top} \cdot k = (M\cdot w)^{\top} \cdot k =0,$$
because $w\in Ker(M)$ and $M\cdot w=0$.   %\qed
\end{proof}

%\begin{proof}
%Let $w$ be in $Ker(M)$ and $v$ in the column space of $M^{\top}$. We denote by $\{c_1, ..., c_m\}$ the set of column vectors of $M^{\top}$. Then, exists a vector $k\in \R^n$ such that $v=<M,k>$ (because $v$ is a linear combination of the column vectors of $M^{\top}$). Now, we have $<w,v>=w^{\top}\cdot v = w^{\top}\cdot <M,k> = w^{\top}\cdot M^{\top} \cdot k = (M\cdot w)^{\top} \cdot k =0$ because $w\in Ker(M)$ and $M\cdot w=0$.   \qed
%\end{proof}

From Lemma \ref{lm1}, a basis of $Row\_Space(M)$ is a basis of the orthogonals of $Ker(M)$. Thus, for 
the square matrix $A$, a vector $v$ is orthogonal to $Ker((A-\l I_n)^n)$, \emph{i.e.}, $<E_\l(A),v>=0$, if an only if $v\in Row\_Space((A-\l I_n)^n)$. We directly deduce the following corollary.

\begin{cor}\label{cor2}
Let $A\in \mathcal{M}_n(\R)$ and $v\neq 0\in \R^n$. The program $P(A,v)$ terminates if and only if for every positive eigenvalue $\l$ of $A$ $v$ is in the vector space $Row\_Space((A-\l I_d)^n)$.\qed
\end{cor}

\begin{proof}
By Lemma \ref{lm1}, the basis of $Row\_Space((A-\l I_d)^n))$ is a basis of the orthogonals of $Ker((A-\l I_d)^n))$. We then apply Theorem \ref{rachidiantheory}.%\qed
\end{proof}

%\begin{proof}
%Using Lemma \ref{lm1}, we know that the basis of $Row\_Space((A-\l I_d)^n))$ is the basis of the orthogonals of $Ker((A-\l I_d)^n))$. We can then use directly Theorem \ref{rachidiantheory} to complete the proof.\qed
%\end{proof}

%FIXME: DO WE NEED TO PUT RUN EX IN A SECTION ?

% !TeX spellcheck=en_US
\section{Running Example}\label{run-ex}

In practice, we can use Corollary \ref{cor2} to support three fast computational steps,
as illustrated in the following example.

\begin{ex}\label{ex2}\emph{(Running example)}
Consider a program $P(A,v)$ where

%\begin{multicols}{2}{
%\begin{small}
%\noindent(i) Pseudo code:
%\begin{lstlisting}
%while(z+t-x-y>0){
%  x := 2x - y;
%  y := -x + 2y -z;
%  z := -y + 2z +t;
%  t := 2t;}
%\end{lstlisting}

%(ii) Associated matrices:\\
$A=\begin{pmatrix} 2 & -1 & 0 & 0\\ -1 & 2 & -1 & 0\\ 0 & -1 & 2 & 1\\ 0 & 0 & 0 & 2\end{pmatrix}$, and $v=\begin{pmatrix} -1\\ -1\\1\\ 1\end{pmatrix}$. 
%\end{small}
%}
%\end{multicols}

\medskip
\noindent\emph{\textbf{Step $1$}: We compute the list $e_{\l}$ of positive eigenvalues for $A$. The result is:}
\begin{small}
\begin{verbatim}
[[2 - sqrt(2), sqrt(2) + 2,2], [1, 1, 2]]
\end{verbatim}
\end{small}
Hence, we have three positive eigenvalues, namely,  $\l_{1}=2, \l_{2}=2-\sqrt{2}, \l_{3}=2+\sqrt{2}$, 
with multiplicities $2$, $1$ and $1$, respectively. 

\medskip
\noindent\emph{\textbf{Step $2$}: We compute the matrix $E_{\l}=(A-\l I_n)^n$ for $\l=2+\sqrt{2}$. The result is:} \\
\begin{small}
\begin{verbatim}
(A - (e[i])*Id_m)^d
[         18  16*sqrt(2)          14  -4*sqrt(2)]
[ 16*sqrt(2)          32  16*sqrt(2)         -14]
[         14  16*sqrt(2)          18 -12*sqrt(2)]
[          0           0           0           4]
\end{verbatim}
\end{small}

\noindent\emph{\textbf{Step $3$}: We check if $v \in Row\_Space(E_{\l})$:}\\
Here we use a standard procedure from linear algebra to check if a given vector belongs to a vector-space spanned by a given set of vectors. We compute the unique \emph{reduced row echelon form} of  
matrix $E_{\l}^{\top}$. For  that we run a \emph{Gaussian elimination} on the rows using the \emph{Gauss-Jordan elimination} algorithm. The generated matrix, below on the left, provides us with a linearly independent basis for $Row\_Space(E_{\l})$.  We remove the rows containing only zero entries, 
and we augment the computed basis with the vector $v^{\top}$ by appending it as the last row. 
We obtain the  matrix below on the right.
\begin{multicols}{2}{
\begin{small}
%\begin{quote}
\begin{verbatim}
(E[i]).echelon_form()
[      1       0      -1       0]
[      0       1 sqrt(2)       0]
[      0       0       0       1]
[      0       0       0       0]  
\end{verbatim}
\end{small}
\vfill
\columnbreak
\begin{small}
\begin{verbatim}
block_matrix([[Er[i]], [V.T]])
[      1       0      -1       0]
[      0       1 sqrt(2)       0]
[      0       0       0       1]
[-------------------------------]
[     -1      -1       1       1]
\end{verbatim}
\end{small}
}
\end{multicols}
Finally, we generate its reduced row echelon form obtaining matrix $R\_S_{\l}$:
\begin{small}
\begin{verbatim}
block_matrix([[Er[i]], [V.T]]).echelon_form()
[1 0 0 0]
[0 1 0 0]
[0 0 1 0]
[-------]
[0 0 0 1]
\end{verbatim}
\end{small}
From the Gauss-Jordan elimination properties, it is well-known that $v$ belongs to the space $Row\_Space(E_{\l})$ if and only if $R\_S_{\l}(n,n+1)=0$. %Also, $R\_S_{\l}(n,n+1)$ is an element of the matrix we have just computed and 
Here we have $R\_S_{\l}(n,n+1)= 1$, which means that $v$ is not in $Row\_Space(E_{\l})$. Thus, by Corollary \ref{cor2}, we conclude that program $P(A,v)$ is nonterminating. 
\qed
\end{ex}

%The necessary and sufficent conditions (see Theorem \ref{rachidiantheory} and its Corollary \ref{cor2}) obtained, allows us to determine the termination of any homogeneous linear program, considering all initial values.% The termination analysis of affine linear programs in $P^{\mathbb{A}}$, reduces to the class of homogeneous linear programs. 
%We could here adapt our technics and use the reduction of the problem to homogeneous linear programs as proposed in \cite{tiwaricav04, Braverman}. 

%FIXME: I moved this paragraph to the section discussion:
%In our recent work in \emph{asymptotically nonterminant values} generation \ref{TR-IC-14-03}, we also provided new and efficient technics that can be adapted to our methods in order to extend the results to linear affine loop programs. its discussion should be raised in an other article with more technical and practical details on the applications of the methods and their experiments.
%===============

%order to adapt our methods to the affine case and its discussion should be raised in an other article with more technical and practical details on the applications of the methods and their experiments. %Thus the presented necessary and sufficient condition provides a decidability result and a complete computational method for determining the termination of the full class of linear/affine programs. 
As we show in Example \ref{ex2}, we avoid the computation of generalized eigenspaces in practice.
Instead, use the exact algorithm associated to Corollary \ref{cor2}. %The following section presents, in more details, our complete approach and its associated algorithm. 

% !TeX spellcheck=en_US
\section{A Complete Procedure to Check Termination}\label{algo}

We use the necessary and sufficient conditions provided by Theorem \ref{rachidiantheory} and its related practical Corollary \ref{cor2} to build a sound and complete procedure to check the termination of linear programs. 
Moreover, the method so obtained is based on few computational steps associated with fast numerical algorithms.  

The pseudo code depicted in Algorithm \ref{Algo-1} illustrates the strategy. 
It takes as input the number of variables, the chosen field where the variables are interpreted, the assignment matrix $A$ and the vector $w$ encoding the loop condition.  
We first compute the list of positive eigenvalues (lines $1$ and $2$ in \ref{Algo-1}).
If this list is empty we can then state that the loop is terminating (lines $3$ and $4$). 
Otherwise, we continue the analysis using the nonempty list of positive eigenvalues.  
For each positive eigenvalues $e'[i]$ we  first need to compute the matrix $E_i=(A-e'[i] I_n)^n$ (line $6$).  Using Corollary \ref{cor2}, we know that the loop is terminating if and only if $w$ is in the $Row\_Space$ of $(A-e'[i] I_n)^n$ for every positive eigenvalue $e'[i]$. 
In other words, for each positive eigenvalue, we have to check if $w$ is in the vector space spanned by the basis of the $Row\_Space$ of the associated matrix $E_i$. 
In order to do so, one first needs to consider the linearly independent vectors $\{r_1, ..., r_n\}$ that form a basis of the $Row\_Space$. This basis is obtained from the list of the non-zero row vectors of the computed \emph{reduced row echelon form} of $E_i$ (lines $7$ and $8$).  
The efficient way to check if $w$ is in the vector space spanned by the basis $\{r_1, ..., r_n\}$ comprises the following computational steps:
%\begin{enumerate}
%\item 
(i) We build the augmented matrix $E_A$ formed by the row vectors $r_1, ..., r_n$ and $w^{\top}$ (line $9$);    
(ii) We compute the \emph{reduced row echelon form} of matrix $E_A$ (line $8$). 
For that we apply \emph{Gaussian elimination} on the rows. This reduced, canonical form is unique and is  computed exactly by the \emph{Gauss-Jordan elimination} method; (iii) We know that the added vector $w$ is in the vector space spanned by $r_1,..., r_n$ if and only if the bottom right entry of the reduced row echelon matrix $E_R$ is null.  
%\end{enumerate}
Thus if $E_R(n,n+1) \neq 0$, we conclude that there exists a positive eigenvalue $e'[i]$ such that $w$ is not in $Row\_Space(A-e'[i]I_n)^n$, which is equivalent to saying that the loop is nonterminating (lines $11$ and line $12$). Otherwise if he have exhausted the list of positive eigenvalues and always found that $w$ is in the $Row\_Space$ of the associated matrix, we conclude that the loop is terminating  (line $13$).

%\begin{center}
\begin{small}
\begin{algorithm}%[htb]
%\dontprintsemicolon
{\bf /*Checking the termination for linear homogeneous programs.*/}\;
\KwData{$n$ the number of program variables, $\mathbb{K}$ the field, $P(A,w)\in P^{\H}$ where $A\in\mathcal{M}(n,\mathbb{K})$ and $w\in\mathcal{M}(n,1,\mathbb{K})$}
\KwResult{Determine the Termination/Nontermination}
\Begin{
\nl $\{e[1],...,e[r]\} \longleftarrow$ {\bf eigenvalues(}$A${\bf)}\;
\nl $\{e'[1],...,e'[s]\} \longleftarrow$ {\bf striclty\_positives(}$\{e[1],...,e[r]\}${\bf)}\;
% \nl  /*$e'[i]$ refers to striclty positive eigenvalues*/\;
\nl \If(there is no positive eigenvalues.) {$\{e'[1],...,e'[s]\}=\emptyset$} 
    {
   \nl    \Return TERMINANT\;
        }
\nl \For{$i=1$ \KwTo $s$}{       
         \nl $\mathbb{E} \longleftarrow$ $(A-e'[i] I_n)^n$\;
         \nl $\mathbb{E}_{rrf} \longleftarrow$ {\bf echelon\_form(}$\mathbb{E}${\bf )}\;
         \nl $\mathbb{E'}_{rrf} \longleftarrow$ {\bf remove\_zero\_row(}$\mathbb{E}_{rrf}${\bf )}\;
         \nl $\mathbb{E}_{A} \longleftarrow$ {\bf augment\_row(}$\mathbb{E'}_{rrf}$, $w^{\top}${\bf )}\; 
         \nl $\mathbb{E}_{R} \longleftarrow$ {\bf echelon\_form(}$\mathbb{E}_{A}${\bf )}\;
         \nl \If {$\mathbb{E}_{R}(n,n+1)\neq 0$} {\nl    \Return NONTERMINANT\;}
    }
%\nl  \If{$i=s \land \mathbb{E}_{e'[i]}(n,n+1)=0$ } 
 %   {
   \nl    \Return TERMINANT\;
  %      } 
}
\caption{  {\bf Termination\_linear\_Loop} $(n,\mathbb{K}, A, w${\bf )}\label{Algo-1}}
\end{algorithm}
\end{small}
%\end{center}

The function {\bf echelon\_form} computes the reduced row echelon form by Guass-Jordan elimination, and its time complexity is of order $O(n^3)$. We interpret the variables in a specified field, \emph{i.e.} an extension of $\Q$, chosen according to the discussion  in Section \ref{count}. 
By using efficient mathematical packages, \emph{e.g.} Maple, \textsf{Mathematica}, \textsf{Sage}, \textsf{Lapack} or \textsf{Eispack}, one can obtain the eigenvalues as closed-form algebraic expressions, \emph{i.e.} the solution of an algebraic equation in terms of its coefficients, relying only on addition, subtraction, multiplication, division, and the extraction of roots. 
Also, it is well known that with $n<5$, the eigenvalues computed by the function {\bf eigenvalues} are already exhibited as such algebraic numbers. 
Moreover, the algorithm for eigenvalue computation has a time complexity that is of order $O(n^3)$, and so the overall time complexity of the algorithm {\bf Termination\_linear\_Loop} remains of the same order
of time complexity.  
%FIXME: say words about the experiments.  

In Table \ref{tab-experim} we list some experimental results. The column \textbf{Set-i} refers to a set of loops generated randomly. %FIXME: Why randomly?
The column \textbf{\#Loops} gives the number of loops treated. %We give the 
%associated fields in column \textbf{Fi}.
%These are the countable subsets described in Section \ref{count}. 
We use the countable subsets described in Section \ref{count}.  
%The ring $\L$ refers to the  countable subset described in Section \ref{count}. 
The column \textbf{Dim} refers to the dimension of the initial systems,
\emph{i.e}, the number of variables. 
The column \textbf{\#T} shows the number of programs found to be terminating, and the column  \textbf{\#NT} gives the number of loop programs found to be non-terminating. 
Finally, column \textbf{CPU/s[T]} refers to  cpu time results while checking all the terminating loop programs, and column \textbf{CPU/s[N]} gives the cpu time taken to check nontermination. 
The  column \textbf{CPU/s[total]} gives  cpu time results, in seconds, for deciding about termination for the given set of $500$ loops. We have implemented our prototype using \textbf{Sage} \cite{sage} with
interfaces written in \textbf{Python}.  By so doing, we were able to access several useful mathematical packages. 
As expected, we can see that  more nonterminating programs were found, as they are easier to write.
Note also that it takes much more time to prove termination than to prove nontermination.

\begin{table}[t]
\caption{Experimental results on randomly generated linear loop programs}\label{tab1} %
\begin{center}
\begin{small}
\begin{tabular}{|c|c|c|c|c|c|c|c|c|}
\hline 
 RandSet&  \#Loops& Dim & \#T & \#NT & CPU/s[T] & CPU/s[N] & CPU/s[total] \\ 
\hline Set-1& $500$  & 3 &$152$ & $348$ &10.02 & 8.79& 18.24\\ 
\hline Set-2& $500$  & 3 &$195$&$305$&$8.97$&$9.11$ & 18.08\\ 
\hline Set-3& $500$  & 3 & $233$ & $267$ & $15.07$ & $12,78$ &27.85 \\ 
\hline Set-4& $500$  & 3 & $223$ & $277$ & 12.49 & 10.42 & 22.91\\ 
\hline Set-5& $500$  & 3 & 246&254 &12.52 &11.59 & 24.11\\ 
\hline Set-6& $500$  & 3 &222 &278 &13.30 &10.35 &23.66 \\   
\hline Set-7& $500$  & 4 & 122&378 &27.8 &16.51 &44.31 \\ 
\hline Set-8& $500$  & 4 &184 &316 &42,67 &21.90 &53.80 \\ 
\hline Set-9& $500$  & 4 &145 &355 &31.91 &18.05 & 49.97\\ 
\hline Set-10& $500$  & 4 &171 &329 &41.16 &22.37 &63.54 \\ 
\hline Set-11& $500$  & 4 &185 &315 &43.03 & 24.22&67.25 \\ 
\hline Set-12& $500$  & 4 & 176& 324&40.36 &19.95 &60.32 \\  
\hline Set-13& $500$  & 5 &183 &317 &126.24 & 66.95& 193.20\\ 
\hline Set-14& $500$  & 5 &227 &273 &155.80 &81.29 &237.10 \\ 
\hline Set-15& $500$  & 5 &178 &322 &103.90 &43.47 &146.57 \\ 
\hline Set-16& $500$  & 5 &161 &339 &169.92 &54.00 & 223.92\\ 
\hline Set-17& $500$  & 5 &174 &326 &171.92 & 66.75& 238.68\\ 
\hline Set-18& $500$  & 5 &158 &342 &174.91 &70.32 & 254.24\\ 
\hline Set-19& $500$  & 6 &141 & 359& 236.0&70.19  & 306.20\\ 
\hline Set-20& $500$  & 6 &173 &327 &387.80 &105.69 &493.50 \\ 
\hline Set-21& $500$  & 6 &192 &308 &342.70 &101.89 & 444.59\\ 
\hline Set-22& $500$  & 6 &188 &312 &352.40 &165.41 &517.81 \\ 
\hline Set-23& $500$  & 6 & 227& 273 &402.71 &174.56 &577.28 \\ 
\hline Set-24& $500$  & 6 & 184 & 316 &385.00 &190.94& 575.94 \\ 
\hline Set-25& $500$  & 7 & 171& 329&851.18 &194.21& 1044.39 \\ 
\hline Set-26& $500$  & 7 & 139& 361&699.03 &174.65& 873.68 \\  
\hline Set-27& $500$  & 7 & 166& 334&876.62 &238.94& 1115.56 \\   
\hline 
\end{tabular}
\end{small}
\label{tab-experim}
%}
\end{center}
\end{table}

%%\input{reductiontoP1}
%%\input{runningexSAS2013}
% !TeX spellcheck=en_US
\section{Variables Over Countable Sets}\label{count}
In this section, we show that to check the termination of a linear program $P(A,v)$ with one loop condition over $\R^n$, we can restrict the analysis 
to the case where the variable belongs to a countable subset of $\R^n$, depending on $A$. First, we study an example, which is already interesting in itself, and which will prove that we cannot restrict the interpretation of the variable over the field $\Q$ of rational numbers if we want to prove the termination for all real inputs. We start with two elements of $\Q(\sqrt{2})-\Q$, which are conjugate under the Galois group $Gal_{\Q}(\Q(\sqrt{2}))$, 
of opposite signs, and the negative one 
of absolute value strictly greater than the positive one. 
For instance, take $\l^-=-1-\sqrt{2}$, and $\l^+=-1+\sqrt{2}$. They are the roots of the polynomial $P(X)=(X-\l^-)(X-\l^+)=X^2+2X-1$.  Now let $A=\begin{pmatrix} 0 & 1 \\ 1 & -2 \end{pmatrix}$ be the associated companion matrix, so that its characteristic polynomial is $P$, and its eigenvalues are $\l^-$ and  $\l^+$. Its generalized eigenspaces are easy to compute.
We find
$E_{\l^-}(A)= \R.e^-$  and $E_{\l^+}(A)= \R.e^+$ with  $e^-=\begin{pmatrix} 1 \\ \l^- \end{pmatrix}$ and $e^+=\begin{pmatrix} 1 \\ \l^+ \end{pmatrix}$.  Now let $v=(1 , 0)^\top$. We have $<v,e^+ >=1$ and so, according to Theorem \ref{rachidiantheory}, the program $P_1=P(A,v)$, associated to $A$ and $v$, does not terminate. We can actually find the points of $\R^2$ for which the program is not terminating.

\begin{prop}\label{locusexample}
Let $A$, $v$ and $P_1$ be as above.
Then program $P_1$ does not terminate for an initial condition $x\in \R^2$ if and only if 
$x\in E_{\l^+}(A)$ and $\langle x,v \rangle >0$, i.e. $x\in \R_{>0}.e^+.$\qed
\end{prop}
\begin{proof}
 If $x=t.e^+$, with $t>0$, then 
 $A^k(x)=t{\l^+}^k.x$, and $<v,A^k(x)>=t{\l^+}^k>0$ 
for all $k\geq 0$.
Hence, the program does not 
terminate with such an $x$ as initial condition.
Conversely, suppose that $x$ satisfies  $\langle v,A^k(x) \rangle >0$ for all $k\geq 0$. Decompose $x$ on the basis $(e^-,e^+)$. Then 
$x=s.e^-+t.e^+$, and $A^k(x)=s{\l^-}^k.e^- +t{\l^+}^k.e^+$, so that 
$<v,A^k(x)>=s{\l^-}^k +t{\l^+}^k$. Suppose that $s$ is not zero. As $|\l^-|>|\l^+|$, for $k$ large enough, the scalar $<v,A^k(x)>$ will be of the same sign as $s{\l^-}^k$, which  alternates positive and negative. Since this is absurd, $s=0$. Now as  $<v,A^k(x)>=t{\l^+}^k$, this implies that $t>0$, and so the 
proposition holds. %\qed
\end{proof}

Proposition \ref{locusexample} leads us to the following corollary.

\begin{cor}\label{counter}
 With $A$ and $v$ as above,  program $P_1$ is terminating on $\Q^2$, but not on $\R^2$\qed
\end{cor}
\begin{proof}
 We already saw that $P_1$ does not terminate on $\R^2$. Now let $x$ be an element of $\Q^2$. If $P_1$ was not terminating with $x$ 
as an initial value, this would imply  that $x$ is in $\R_{>0}.e^+$, according to Lemma \ref{locusexample}. However, no element of 
$\Q^2$ belongs to $\R_{>0}.e^+$ because the quotient of the coordinates of $e^+$ is irrational. This implies that $P_1$ terminates on 
$\Q^2$.%\qed
\end{proof}

This proves that even if $A$ and $v$ are rational, one cannot guarantee the termination over the reals if the interpretation of the variables are restricted to rationals. It is clear that one cannot hope to produce 
any valid conjecture of this type if $A$ and $v$ have wild coefficients,
like transcendentals, for example. 
However, when $A$ and $v$ have algebraic coefficients, using Corollary \ref{cor2}, one can find a simple remedy. It is indeed enough to replace $\Q$ by a finite extension of the field $\Q$. Such an extension $K$ is called a \textit{number field}, and is known to be countable.
Indeed, it is a $\Q$-vector space of finite dimension, \emph{i.e.}, $K=\Q.k_1\oplus \dots \oplus \Q.k_l$ for some $l\geq 1$, and elements $k_i$ in $K$.
It is, moreover, known that $K$ is the fraction field of its \textit{ring of integers} $O_K$, which is a free $\Z$-module of finite type. 
In fact $O_K=\Z.o_1\oplus \dots \oplus \Z.o_l$ for the same $l\geq 1$, and where 
the elements $o_i$ can be chosen equal to the $k_i$, 
for well chosen $k_i$'s. We say that a number field is \textit{real} if it is a subfield of $\R$.
Notice that in the mathematical literature 
a totally real number field is a number field with only real embeddings in $\C$.
Here what we call real is thus weaker than totally real.  

\begin{thm}\label{testcountable}
 Let $A\in \mathcal{M}_n(\R)$, $v\neq 0\in \R^n$, and suppose that their coefficients are actually in $\Q$ or, more generally, in a \emph{real} 
number field $K$. 
Then there is a well-determined \emph{real} finite extension $L$ of $\Q$, or of $K$ in the general case, which is contained in $\R$ and such that the program $P(A,v)$, 
associated to $A$ and $v$, terminates if and only if it terminates on the countable set $L^n$. We can choose $L$ to be the extension $\Q(\l_1,\dots,\l_t)$ of $\Q$, or $K(\l_1,\dots,\l_t)$ in general, spanned by the 
positive eigenvalues $(\l_1,\dots,\l_t)$ of $A$. It is actually enough 
to check the termination of the program on $O_L^n$.\qed
\end{thm}

\begin{proof} 
We deal with the general case.
The reader not familiar with field extensions can just replace $K$ by $\Q$.
It is obvious that if the program terminates, it terminates on $L^n$ for any subset $L$ of $\R$. Now let $\l_1,\dots,\l_r$ be 
the positive eigenvalues of $A$. They are all roots of the minimal (or characteristic) polynomial $Q$ of $A$, which is in $K[X]$.
Hence they are all algebraic on $K$, and so also on $\Q$ as $K/\Q$ is finite.  
Let $L=K(\l_1,\dots,\l_r)\subset \R$. Suppose that the program $P_1$ does not terminate. 
Then there is some $i\in \{1,\dots,r\}$, such that $<E_{\l_i},v>\neq 0$ according to Corollary \ref{rachidiantheory}. Let $r$ be the positive integer such that 
$Ker((A-\l_i I_n)^r)\not\subset v^{\perp}$, but 
$Ker((A-\l_i I_n)^{r-1}) \subset v^{\perp}$.
As in the proof of Theorem \ref{endo-thm},  for any $x$ 
in $Ker((A-\l_i I_n)^r)-Ker((A-\l_i I_n)^{r-1})$, such that $<v,x>>0$, the program does not terminate. 
We fix such an $x$. 
Since both spaces $Ker((A-\l_i I_n)^r)$ and $Ker((A-\l_i I_n)^{r-1})$ are defined by linear equations with coefficients in $L$,  there is a basis of 
$Ker((A-\l_i I_n)^r)$ with coefficients in $L^n$ containing a basis of $Ker((A-\l_i I_n)^{r-1})$ with coefficients in $L^n$. 
It is easy to see that this  implies that $L^n \cap [Ker((A-\l_i I_n)^r)-Ker((A-\l_i I_n)^{r-1})]$ is dense in 
$Ker((A-\l_i I_n)^r)-Ker((A-\l_i I_n)^{r-1})$, because $L$ contains $\Q$ which is dense in $\R$.  Hence, there is a sequence $x_k$ in $L^n \cap [Ker((A-\l_i I_n)^r)-Ker((A-\l_i I_n)^{r-1})]$ which approaches $x$.
In particular, $\langle v,x_k\rangle>0$ for $k$ large enough. 
Thus the program  does not terminate on $x_k$ when $k$ is such that $\langle v,x_k \rangle >0$. This shows that $P_1$ 
does not terminate on $L^n$. 
The fact that $P_1$ does not terminate on $O_L$ is a trivial consequence of the fact that 
any element of $L$ is the quotient of two elements of $O_l$.
In particular, if $P_1$ does not terminate on $x\in L^n$, take $a>0$ in $O_L$ such that 
$ax\in O_L^n$.
Then the program does not terminate on $ax$.%\qed 
\end{proof}

We now show how Theorem \ref{testcountable} applies on our previous example.

\begin{ex}
 For the program associated to  matrix $A=\begin{pmatrix} 0 & 1 \\ 1 & -2 \end{pmatrix}$ and  vector 
$v=(1 , 0)^{\top}$, we get $L=\Q(\l^+)=\Q(\sqrt{2})=\{a+b\sqrt{2}:a\in \Q, b\in \Q\}$. 
Its ring of integers is $O_L=\Z(\l^+)=\Z(\sqrt{2})=\{a+b\sqrt{2}:a\in \Z, b\in \Z\}$. Theorem \ref{testcountable} asserts that, as the program $P(A,v)$ is non terminating, it is already non terminating on $O_L^2$. Indeed, 
take $x^+$ as an initial value, then $x^+= \begin{pmatrix} 1 \\ -1+\sqrt{2} \end{pmatrix}$ is in $O_L^2$, and we saw that 
$P(A,v)$ does not terminate on $x^+$.\qed
\end{ex}

% !TeX spellcheck=en_US
\section{Discussion}\label{discus}

The important papers \cite{Braverman, tiwaricav04}, treating
homogeneous linear programs, can be seen, at first, as closely related
to our results.  The sufficient condition fully proved and established
as our preliminary results in section \ref{sufficient}, was first
stated in \cite{tiwaricav04}. On the other hand, the sufficient
conditions proposed in \cite{Braverman, tiwaricav04} are not necessary
conditions for the termination of homogeneous linear programs and,
thus, it is not obvious that one can obtain from those results a direct encoding leading deterministically to a practical algorithm.  
The treatment in \cite{tiwaricav04} can be
divided in two parts. First, the interesting sketch of the proof for
the sufficient condition leaves space for elaboration. We completed it
in a solid mathematical way. We found obstacles that were not obvious how to circumvent. Like applying Brouwer's fixed point theorem to appropriate spaces, and having $0$ in the closure
of the orbit of a variable under the action of the transition matrix. 
The second part provides a lengthy procedure to
check for termination.  
It comprises $3$ reductions, a case
analysis, and long and costly symbolic computations (whose complexities are $max(O(n^6), O(n^{m+3}))$ where $n$ is the number of variables and $m$ the number of conditions \cite{xia2011}.).
Also ideas presented in  \cite{Braverman} are based on the approach
proposed in \cite{tiwaricav04}, while considering termination analysis over the integers.  Similar points could be raised concerning the work in  \cite{Braverman},
that is, a  
complex procedure is proposed, one that appears lengthy and costly.  
In fact, it is not
clear to us if those approaches give rise to simple and fast
algorithms. Instead, we have a more direct and clear statement which naturally
provides a simple algorithm to check termination, as illustrated by our
examples, and with much better complexity.  Moreover, we show that it is
enough to interpret the variable values over a countable number field,
or over its ring of integers, in order to determine program termination
over the reals. 

In a recent work about \emph{asymptotically nonterminating values}
($ANT$) generation \cite{TR-IC-14-03}, we also provided new and
efficient techniques to extend our results to general affine loop
programs, \emph{i.e.}, programs with several loop conditions.  
We defer this discussion to another companion article, where more practical
details will be presented, together with some experiments.  
%Classical termination problems consider any possible initial values as done in
%the related literature. If our procedure returns terminating on any
%arbitrary initial values one can obviously have the same conclusion
%considering any initial precondition.  The contributions presented in
%this paper are central to static data input analysis that we developed
%in those recent works. 

The generated $ANT$ set can be used directly as
preconditions for termination or it can be intersected efficiently with
another given preconditions, provided by other static analysis methods
for instance.

Our main results, Theorem \ref{rachidiantheory} and its
Corollary \ref{cor2}, with a direct encoding as in Algorithm
\ref{Algo-1}, together with  the results in Section \ref{count},
guaranteeing the symbolic computation while circumventing rounding
errors, are evidences of the novelty of our approach.

%%\input{experimentSAS2013}
% !TeX spellcheck=en_US
\section{Conclusions}\label{conclusion}
We presented the \emph{first necessary and sufficient condition} for the termination of linear homogeneous loop programs. This condition leads to a sound and complete procedure for checking termination for this class of programs. 
The analysis of the associated algorithms shows  that the new method operates in fewer computational steps than all known routines that support the mathematical foundations of previous methods. 
Section \ref{count}, and especially the example therein, introduces the important notion of the locus of initial variables values for which a linear program terminates. 
In that example, it allowed us to decide if the program terminates on all rational
initial variables values. 
Actually, these  methods can be vastly generalized in order to treat 
the termination problem for linear programs on rational initial values. 
However, we suspect  that this development it will involve some Galois theory, as well as our results 
on asymptotically non terminating variable values \cite{TR-IC-14-03}, and so we 
prefer to pursue this investigation in  
the near future.

\bibliographystyle{splncs}
\bibliography{termination}

\end{document}